\documentclass[sigconf]{aamas}  

\usepackage{amsmath, amsthm, latexsym, graphics, graphicx,amsfonts, color}
\usepackage{todonotes} 
\usepackage{verbatim} 

\usepackage{graphics,graphicx}
\usepackage{epsf}
\usepackage{stmaryrd}
\usepackage{longtable}
\usepackage{latexsym}
\usepackage{amssymb}
\usepackage{color}
\usepackage[mathscr]{eucal}
\usepackage[refpage,english]{nomencl}
\usepackage{xspace}
\usepackage{algorithmic} 
\usepackage{algorithm}
\usepackage{subfigure}
\usepackage{wrapfig}
\usepackage{url}
\usepackage{mathtools}
\usepackage{tikz}
\usepackage{balance}

\makeatletter
\DeclareRobustCommand*\cal{\@fontswitch\relax\mathcal}
\makeatother

\newcommand{\Ag}{{\cal A}} 

\newcommand{\GA}{{\Gamma}} 

\def\solveApprox{SApp}
\def\ModelApprox {MApprox}
\newcommand{\varM}[1][]{V_{m_{_{#1}}}} 
\newcommand{\bitM}[1][]{v_{M_{_{#1}}}} 
\def\MAover {v_{M^{\Ag}_{over}}} 
\def\MGover {v_{M^{\GA}_{over}}} 
\def\MAunder {v_{M^{\Ag}_{under}}} 

\let\none = -

\def\F {{\rm{F}}}

\def\PV {{\mathcal{PV}}}

\def\G {{\rm G}}  
\def\X {{\rm X}}  
\def\F {{\rm F}}  

\def\CTL {{${\rm CTL}$}}

\def\ATL {{${\rm ATL}$}}

\newcommand{\coop}[1]{\langle\!\langle{#1}\rangle\!\rangle}
\newcommand{\set}[1]{\{{#1}\}}
\newcommand{\extended}[1]{}
\newcommand{\stan}{s} 
\newcommand{\States}{{\mathcal S}t}

\newcommand{\seq}{\pi}

\def\trans#1{\stackrel{#1}{\longrightarrow}} 

\newcommand{\prop}[1]{\ensuremath{\mathsf{{#1}}}}
\newcommand{\Next}[1][]{X_{{#1}}  
        \,}

\newcommand{\Always}[1][]{G_{{#1}}}
\newcommand{\Until}[1][]{U_{{#1}}}

\newcommand{\Epath}{\mathsf{E}}
\newcommand{\Apath}{\mathsf{A}}

\newcommand{\strat}{\sigma}

\newcommand{\CollStra}{\Sigma_{\GA}} 

\newcommand{\stratstyle}[1]{\ensuremath{\mathrm{#1}}}
\newcommand{\ir}{\stratstyle{ir}\xspace}
\newcommand{\iR}{\stratstyle{iR}\xspace}
\newcommand{\Ir}{\stratstyle{Ir}\xspace}
\newcommand{\IR}{\stratstyle{IR}\xspace}

\newcommand{\defeq}{\coloneqq}
\newcommand{\outcome}{\mathit{out}}
\newcommand{\grandn}{{\mathbb N}}

\newcommand{\propp}{\mathsf{p}}
\newcommand{\model}{M}
\newcommand{\satisf}[1][]{\models} 

\def\mk#1{\textcolor{blue}{#1}}
\def\wp#1{\textcolor{violet}{#1}}

\def\AG {AG}

\begin{document}

\title{SAT-Based ATL Satisfiability Checking}  



\author{Magdalena Kacprzak}
\affiliation{Bialystok University of Technology \city{Bialystok, Poland}}
\email{m.kacprzak@pb.edu.pl}

\author{Artur Niewiadomski}
\affiliation{\institution{Siedlce University, Faculty of Exact and
Natural Sciences} \city{Siedlce, Poland}}
\email{artur.niewiadomski@uph.edu.pl}

\author{Wojciech Penczek}
\affiliation{ICS PAS, Warsaw, Poland}
\email{penczek@ipipan.waw.pl}

\begin{abstract}
Synthesis of models and strategies is a very important problem in software engineering.
The main element here is checking the satisfiability of formulae expressing the specification
of a system to be implemented.
This paper puts forward a novel method for deciding the satisfiability of formulae
of Alternating-time Temporal Logic (ATL).
The method presented expands on one for CTL exploiting SAT Modulo Monotonic Theories solvers.
Similarly to the CTL case, our approach appears to be very efficient.
The experimental results show that we can quickly test the satisfiability of large ATL formulae
that have been out of reach of the existing approaches.
\end{abstract}


\maketitle


\section{Introduction}




The problem of synthesis is a very important issue in the rapidly-growing field of artificial intelligence
and modern software engineering \cite{Jones2012Group,Kouvaros2018,Schewe2007Distributed}.
The aim is to automatically develop highly innovative software, also for AI robots, chatbots or autonomous
self-driving vehicles.
The problem consists in finding a model satisfying a given property, provided the property is satisfiable.
Finally, the model is transformed into its correct implementation.


A convenient formalism to specify the game-like interaction between processes in distributed systems
is Alternating-Time Temporal Logic (ATL) \cite{ATL97,Bulling2015Logics}.
The interpretation of ATL formulae uses the paradigm of multi-agent systems and is defined in models
like concurrent game structures or interpreted systems.
This logic was introduced to reason about the strategic abilities of agents and their groups.
The strategic modalities allow for expressing the ability of agents to force their preferences
or to achieve a desired goal and are therefore suitable for describing properties like the existence of a winning strategy.
This is particularly important when we study properties and verify the correctness of security protocols or voting systems.
There are a lot of papers analysing different versions of
ATL \cite{Belardinelli2019,Dima11undecidable,Guelev11atl-distrknowldge,Schobbens04ATL,Bulling10verification,Bulling11mu-ijcai,Dima14mucalc,Jamroga2017,Dima15fallmu}
and other modal logics of strategic ability \cite{Chatterjee2010Strategy,Vardi2012Decidable,MogaveroMPV2012What}.
However, there is still a need for developing and introducing new and innovative techniques for solving synthesis
and satisfiability problems \cite{bloem2014sat,Bloem2012SynthesisOR,finkbeiner2013bounded,kupferman2005safraless,Pnueli:1989:Synthesis}.
This is because these problems are hard and their solutions require searching for effective practical algorithms.

\subsection{Contribution}

In this paper we:
\begin{itemize}
\item introduce a novel technique for checking ATL satisfiability, applying for the first time SAT Modulo Monotonic Theories solvers,


\item propose a method which is universal in the sense that it can be extended to different classes of multi-agent systems
      and ATL under different semantics, 

\item propose a method which allows for testing satisfiability in the class of models that meet given restrictions,

\item present a new efficient tool for checking satisfiability of ATL. 

\end{itemize}

\subsection{Related Work}

The complexity of the ATL satisfiability problem was proven to be EXPTIME-complete
by van Drimmelen \cite{van2003satisfiability,goranko2006complete} for a fixed number of agents,
and by Walther et al. \cite{walther2006atl} for systems without this assumption.
The satisfiability of ATL$^*$ was proved to be 2EXPTIME-complete \cite{schewe2008atl}.
A method for testing the satisfiability of ATL was developed by Goranko and Shkatov \cite{GorankoShkatov2009}.
Subsequently, this method was extended for checking ATL$^*$ \cite{david2015deciding} and ATEL \cite{Belardinelli2014Reasoning}.




In this paper we propose a solution to the first stage of the synthesis problem, which consists in finding a model
for a given ATL formula.
For this purpose, we adopt the method based on SAT Modulo Monotonic Theories (SMMT) \cite{Fast-flexible-CTL-2016}
used to search for models of the CTL formulae.
This technique was introduced by Bayless et al. in \cite{Bayless-sat-modulo-2015} for building efficient lazy SMT solvers
for Boolean monotonic theories.
Next, Klenze et al. in \cite{Fast-flexible-CTL-2016} presented how the SMMT framework can be used to build
an SMT solver for CTL model checking theory, and how to perform efficient and scalable CTL synthesis.

In this paper we go one step further by developing an SMMT solver for ATL formulae and show how to construct,
often minimal, models for them.
We compare the experimental results with the only implementation of the tool for testing ATL satisfiability described
in the literature \cite{david2015deciding}.
In that paper, unlike in our work, concurrent game structures were used as models for ATL* with perfect recall
and perfect knowledge semantics.

The main advantage of our framework consists in the promising preliminary experimental results and the fact that we can test
satisfiability in classes of models under given restrictions on the number of agents, their local states, transition functions,
local protocols, and valuation of variables.
Restrictions on the number of agents and their local states result directly from the finite model property for ATL \cite{goranko2006complete}.
In addition, it is possible to extend our approach to testing different classes of models and different types of strategies.

\subsection{Outline}

In Sec. \ref{section:MAS-and-ATL} 
we define a multi-agent system and its model, and give the syntax and semantics of ATL.
Sec. \ref{section:ATL-model-checking} defines Boolean monotonic theory for ATL.
In Sec. \ref{section:ATL-formula-approximation} the approximation algorithm is given and its properties are proved.
Sec. \ref{section:SATandSynthesis} introduces the algorithm for deciding ATL satisfiability and model construction.
Sec. \ref{section:results} presents experimental results.
Conclusions are in Sec. \ref{section:conclusions}.

\section{MAS and ATL}
\label{section:MAS-and-ATL}

Alur et al. introduced ATL logic taking into account different model compositions of open systems
like turn-based, synchronous, asynchronous, with fairness constraints or Moore game structures.
In this paper we follow Moore synchronous models \cite{ATL02}, i.e., assume that the state space
is the product of local state spaces, one for each agent, all agents proceed simultaneously,
and each agent chooses its next local state independently of the moves of the other players.
This is a restricted class of models, but it allows for 
the efficient testing of ATL satisfiability.

\subsection{Multi-agent System}
\label{subsection:MAS-and-IS}

We start with defining a multi-agent system following \cite{ATL02,Partial-order-reductions-2018}.


\begin{definition}
\label{def:MAS}
A \emph{multi-agent system (MAS)} consists of $n$ agents
${\cal A} = \set{1,\dots,n}$\footnote{The environment component may be added here with no technical difficulty.},
where each agent $i \in \cal A$ is associated with a 5-tuple
$\AG_i = \left( L_i, \iota_i, Act_i, P_i, T_i \right)$ including:

\begin{itemize}
\item a set of \emph{\extended{possible} local states}
$L_i=\{l_i^1, l_i^2,\dots,l_i^{n_i}\}$;

\item an \emph{initial local state} $\iota_i\in L_i$;

\item a set of \emph{local actions} $Act_i=\{\epsilon_i,a_i^1,a_i^2,\ldots,a_i^{m_i}\}$;


\item a \emph{local protocol} $P_i: L_i \to 2^{Act_i}$ which
selects the actions available at each local state; we assume that
$P_i(l_i) \neq \emptyset$ for every $l_i \in L_i$;

\item a (partial) \emph{local transition function} $T_i: L_i
\times Act_i \rightarrow L_i$ such that $T_i(l_i,a)$ is defined
iff $a\in P_i(l_i)$ and $T_i(l_i,\epsilon_i) = l_i$
whenever $\epsilon_i \in P_i(l_i)$ for each $l_i \in L_i$.

\end{itemize}
\end{definition}

An example MAS specification 
is depicted in Figure~\ref{fig:ex1},
where
${\cal A} = \set{1,2}$,
$\AG_1 = \left(\set{l_1^1, l_1^2, l_1^3}, l_1^1, \set{\epsilon_1, a_1^1, a_1^2, a_1^3}, P_1, T_1 \right)$,\\
$P_1 = \left\{ \left(l_1^1, \{ a_1^1, a_1^3\}\right),\left(l_1^2, \{ a_1^2\}\right),\left(l_1^3, \{a_1^2, a_1^3\}\right) \right\}$,
$T_1 = \bigl\{$
    $\bigl( (l_1^1, a_1^1), l_1^1 \bigr)$,
    $\bigl( (l_1^1, a_1^3), l_1^3 \bigr)$,
    $\bigl( (l_1^2, a_1^2), l_1^2 \bigr)$,
    $\bigl( (l_1^3, a_1^2), l_1^2 \bigr)$,
    $\bigl( (l_1^3, a_1^3), l_1^3 \bigr)$
$\bigr\}$,\\
$\AG_2 \!=\! \left(\!\set{l_2^1, l_2^2}, l_2^1,\! \set{\epsilon_2, a_2^1, a_2^2}, P_2, T_2 \!\right)$,
$P_2 = \bigl\{\! \left(l_2^1, \{ a_2^1, a_2^2\}\right)$,
$\left(l_2^2, \{ a_2^2\}\right)\! \bigr\}$,
$T_2 = \bigl\{$
    $\bigl( (l_2^1, a_2^1), l_2^1 \bigr)$,
    $\bigl( (l_2^1, a_2^2), l_2^2 \bigr)$,
    $\bigl( (l_2^2, a_2^2), l_2^2 \bigr)$
    $\bigr\}$.

\begin{figure}
    \centering
        \includegraphics[width=.375\columnwidth,clip=true,trim=5.25cm 16.15cm 12.3cm 6.1cm]{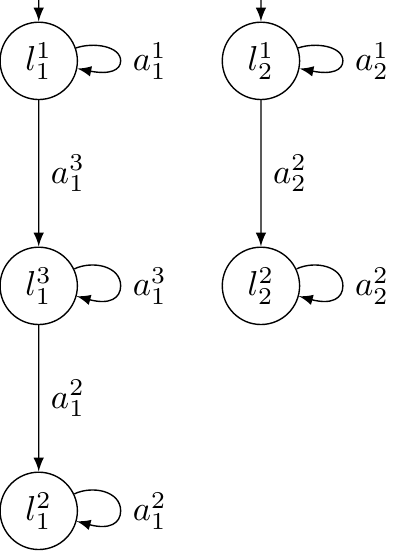}%
    \caption{Visualization of an example MAS specification for ${\cal A}=\{1,2\}$. On the left $\AG_1$, on the right $\AG_2$.}
    \label{fig:ex1}
\end{figure}

In our approach we consider \textit{synchronous} multi-agent systems, i.e.,
systems in which each \textit{global action} is a $n$-tuple $(a^1,\dots,a^n)$,
where $a^i \in Act_i$, i.e., each agent performs one local action.
Define the set of all global actions as $Act = Act_1 \times \dots \times Act_n$.

In order to describe the interaction between agents, the model for $MAS$ is defined formally below.

\begin{definition}[Model]
\label{def:IS}


Let $\PV$ be a set of propositional variables and $MAS$ be a multi-agent system with $n$ agents.
An (induced) \emph{model}, is a 4-tuple $M=(\States,\iota,T,V)$ with

\begin{itemize}
\item the set $\States = L_1\times\dots\times L_n$ of the global\textit{
states},

\item an \textit{initial state} $\iota=(\iota_1,\dots,\iota_n) \in
\States$,

\item the \textit{global transition function} $T : \States\times
Act
\rightarrow \States$, such that $T(\stan_1,a)= \stan_2$ iff 
$T_i(\stan_1^i,a^i) = \stan^i_2$ for all $i \in {\mathcal A}$, where for global state
$\stan = (l_1, \dots, l_n)$ we denote the local component of agent $i$ by $\stan^i=l_i$
and for a global action $a = (a^1, \dots, a^n)$ we denote the local action
of agent $i$ by $a^i$;

\item a valuation of the propositional variables $V: \States
\rightarrow 2^{\PV}$.
\end{itemize}

\end{definition}

\begin{figure}[tb]
    \centering
        \includegraphics[width=\columnwidth,clip=true,trim=5.25cm 10.4cm 4.7cm 5.9cm]{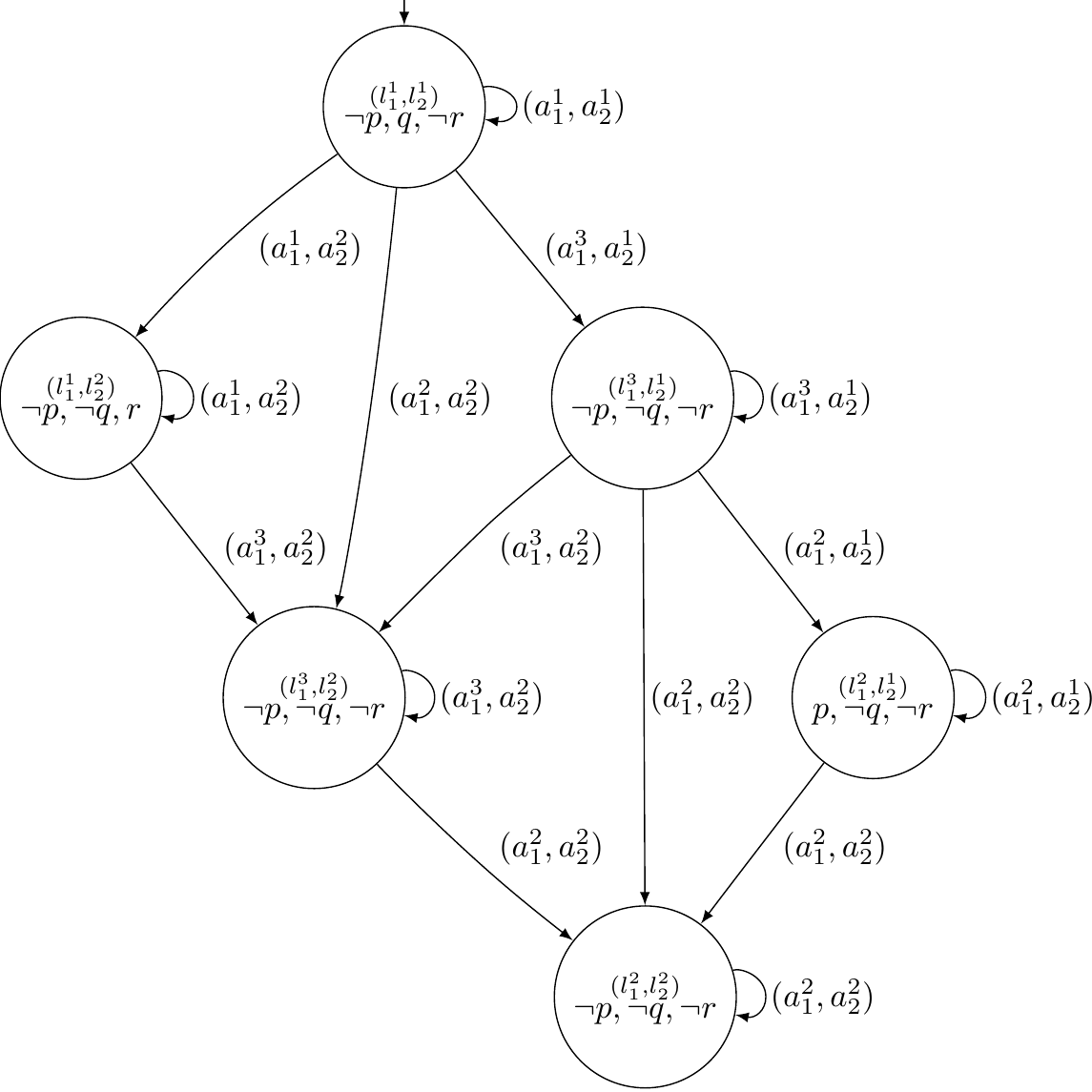}%
    \caption{The model for MAS specification of Fig.~\ref{fig:ex1} and $\PV=\{p,q,r\}$.}
    \label{fig:ex2}
\end{figure}

We
say that action $a \in Act$ is \emph{enabled} at $\stan\in
\States$ if $T(\stan, a) = \stan'$
for some $\stan' \in \States$.
We assume that at
each $\stan\in \States$ there exists at least one enabled action, i.e.,
for all $\stan \in \States$ exist $a \in Act$, $\stan' \in \States$, such that $T(\stan, a) = \stan'$.
%
An infinite sequence of global states and actions $\seq = \stan_0
a_0 \stan_1 a_1 \stan_2\dots$ is called a\extended{n(interleaved)}
\emph{path} \extended{if there is a sequence of
global transitions from $\stan_0$ onwards, i.e., }if
$T(\stan_i, a_i) = \stan_{i+1}$
for every $i \geq 0$.
Let $Act(\seq) = a_0a_1a_2\ldots$ be the sequence of actions in $\seq$,
and $\seq[i] = \stan_i$ be the $i$-th global state of $\seq$.
$\Pi_M(\stan)$ denotes the set of all paths in $M$ starting at
$\stan$.

\subsection{Alternating-time Temporal Logic}
\label{subsection:ATL}

\emph{Alternating-time temporal logic}, \ATL~
\cite{ATL97,ATL98,ATL02} generalizes the bran-ching-time temporal logic
\CTL~\cite{Clarke81ctl} by replacing the path quantifiers
$\Epath,\Apath$ with \emph{strategic modalities} $\coop{\GA}$.
Informally, $\coop{\GA}\gamma$ expresses that the group of agents
$\GA$ has a collective strategy to enforce the temporal property
$\gamma$. The formulae make use of temporal operators: ``$\Next$''
(``next''), ``$\Always$'' (``always from now on''), $\Until$
(``strong until'').


\begin{definition}[Syntax of \ATL]
In vanilla \ATL, every occurrence of a strategic modality is
immediately followed by a temporal operator. Formally, the
language of \ATL\ is defined by the following grammar:
%
$\varphi ::= \prop{p}  \mid  \lnot \varphi  \mid
\varphi\land\varphi  \mid
  \coop{\GA} \Next\varphi  \mid  \coop{\GA} \varphi\Until\varphi \mid \coop{\GA} \Always \varphi$.
\end{definition}
Let $M$ be a model. A \emph{strategy} of agent $i \in {\mathcal
A}$ in $M$ is a conditional plan that specifies what $i$ is going
to do in any potential situation. \extended{A number of semantic
variations are possible.}

In this paper we focus on memoryless perfect information
strategies. Formally, a \emph{memoryless perfect information
strategy for agent $i$} is a function $\strat_i \colon \States \to
Act_i$ st. $\strat_i(\stan) \in P_i(\stan^i)$ for each\extended{
global state} $\stan \in \States$.

A \emph{joint strategy} $\strat_{\GA}$ for a coalition $\GA
\subseteq {\mathcal A}$ is a tuple of strategies, one per agent $i
\in \GA$. We denote the set of $\GA$'s collective memoryless
perfect information strategies by $\CollStra$.

Additionally, let $\strat_{\GA} = (\strat_1,\ldots,\strat_k)$ be a
joint strategy for $\GA = \{i_1,\ldots,i_k\}$. For each $\stan
\in \States$, we define $\strat_{\GA}(\stan) \defeq
(\strat_1(\stan),\ldots,\strat_k(\stan))$.


\begin{definition}[Outcome paths]\label{def:outcome}
The \emph{outcome} of strategy $\strat_{\GA} \in \CollStra$ in
state $\stan \in \States$ is the set
$\outcome_M(\stan,\strat_{\GA}) \subseteq \Pi_M(\stan)$ s. t.
$\seq = \stan_0 a_0 \stan_1 a_1 \dots \in
\outcome_M(\stan,\strat_{\GA})$ iff $\stan_0 = \stan$ and $\forall
i \in \grandn$
$\forall j \in {\GA}$, $a_i^j = \strat_j(\seq[i])$.
\end{definition}
%
%
Intuitively, the outcome of a joint strategy $\strat_{\GA}$ in a global state $\stan$ is the set of all the infinite paths that can occur when
in each state of the paths agents (an agent) in
${\GA}$ execute(s) an action according to $\strat_{\GA}$ and
agents (an agent) in ${\mathcal A} \setminus {\GA}$ execute(s) an
action following their protocols.

The semantics of \ATL\
is defined as follows:

\begin{description}
\item[{$\model,\stan \satisf[Y] \propp$}] iff  $\propp \in
V(\stan)$, for $\propp \in \PV$;

\item[{$\model,\stan \satisf[Y] \neg\varphi$}] iff $\model,\stan
\not\satisf[Y]  \varphi$;

\item[{$\model,\stan \satisf[Y] \varphi_1\land\varphi_2$}]
  iff $\model,\stan \satisf[Y] \varphi_1$ and $\model,\stan \satisf[Y] \varphi_2$;

\smallskip
\item[{$\model,\stan \satisf[Y] \coop{{\GA}}\Next \varphi$}] iff
there is a strategy $\strat_{\GA}\in \CollStra$ such that

$\outcome_\model(\stan,\strat_{\GA}) \neq \emptyset$ and, for each
path $\seq\in \outcome_\model(\stan,\strat_{\GA})$, we have
$\model,\seq \satisf[Y] \Next \varphi$, i.e., $\model,\seq[1]
\satisf[Y] \varphi$;

\smallskip
\item[{$\model,\stan \satisf[Y] \coop{{\GA}} \varphi_1 \Until
\varphi_2$}] iff there is a strategy
$\strat_{\GA}\in \CollStra$ such that

$\outcome_\model(\stan,\strat_{\GA}) \neq \emptyset$ and, for each
path $\seq\in \outcome_\model(\stan,\strat_{\GA})$, we have
$\model,\seq \satisf[Y] \varphi_1 \Until \varphi_2$, i.e.,
$\model,\seq[i]\ \satisf[Y]\ \varphi_2$ for some $i\ge 0$ and
$\model, \seq[j] \satisf[Y] \varphi_1$ for all $0\leq j< i$;

\smallskip
\item[{$\model,\stan \satisf[Y] \coop{{\GA}} \Always \varphi$}]
iff there is a strategy $\strat_{\GA}\in \CollStra$ such that

$\outcome_\model(\stan,\strat_{\GA}) \neq \emptyset$ and, for each
path $\seq\in \outcome_\model(\stan,\strat_{\GA})$, we have
$\model,\seq \satisf[Y] \Always \varphi$, i.e.,
$\model,\seq[i] \satisf[Y] \varphi$, for every $i \geq 0$.
\end{description}
We omit the $\model$ symbol if it is clear which model is intended.

\begin{definition}{\bf (Validity) }$\;$
An \ATL\ formula $\varphi$ is valid in $M$ (denoted $M \models
\varphi$) iff $M, \iota \models \varphi$, i.e., $\varphi$ is true
at the initial state of the model $M$.
\end{definition}
An example ATL formula, which is satisfied by the model depicted in Figure~\ref{fig:ex2},
is as follows:
%
$\coop{1,2} \F (p \land \neg q \land \neg r)  \land \coop{1} \F ( \neg p \land q \land \neg r) \land \coop{1,2} \X (\neg p \land \neg q \land r)$,
where $\coop{\GA} \F \alpha$ is a short for $\coop{\GA}(true \Until \alpha)$.

\section{Boolean Monotonic Theory for ATL}
\label{section:ATL-model-checking}

In this section we show how to construct a Boolean monotonic theory for ATL,
which allows for building a lazy SMT solver \cite{Bayless-sat-modulo-2015} for ATL.
The resulting tool, a SAT modulo ATL solver, can be used for testing the satisfiability
of the ATL formulae as well as for performing efficient and scalable synthesis.

\subsection{Boolean Monotonic Theory}
\label{subsection:boolean-monothonic-theory}

Consider a predicate $P: \{0, 1\}^n \mapsto \{0, 1\}$.
We say that $P$ is Boolean \textit{positive monotonic}
iff $P(s_1,\dots,s_{i-1}, 0, s_{i+1},\dots,s_n) = 1$ implies
    $P(s_1,\dots,s_{i-1}, 1, s_{i+1},\dots,s_n) = 1$, for all $1\leq i \leq n$.
$P$ is called Boolean \textit{negative monotonic}
iff $P(s_1,\dots,s_{i-1}, 1, s_{i+1},\dots,s_n) = 1$ implies
    $P(s_1,\dots,s_{i-1}, 0, s_{i+1},\dots,s_n) = 1$, for all $1\leq i \leq n$.

The definition of (positive and negative Boolean) monotonicity for a function
$F: \{0, 1\}^n \mapsto 2^S$ (for some set $S$) is analogous.
$F$ is Boolean \textit{positive monotonic} iff
$F(s_1,\dots,s_{i-1}, 0, s_{i+1},\dots,s_n) \subseteq
 F(s_1,\dots,s_{i-1}, 1, s_{i+1},\dots,s_n)$, for all $1 \leq i \leq n$.
A function $F$ is Boolean \textit{negative monotonic} iff
$F(s_1,\dots,s_{i-1}, 1, s_{i+1},\dots,s_n) \subseteq
 F(s_1,\dots,s_{i-1}, 0, s_{i+1},\dots,s_n)$, for all $1 \leq i \leq n$.
In what follows we refer to Boolean monotonicity simply as to monotonicity.

\begin{definition}[Boolean Monotonic Theory]
\label{def:monotonic-theory}

A theory $T$ with a signature $\Omega = (S,S_f,S_r,ar)$, where
$S$ is a non-empty set of elements called sorts or types,
$S_f$ is a set of function symbols,
$S_r$ is a set of relation symbols, and
$ar$ is arity of the relation and function symbols,
is (Boolean) monotonic iff:

\begin{enumerate}
\item the only sort in $\Omega$ is Boolean;

\item all predicates and functions in $\Omega$ are monotonic.
\end{enumerate}
\end{definition}


%
The authors of \cite{Bayless-sat-modulo-2015} introduced techniques for building
an efficient SMT solver for Boolean monotonic theories (SMMT).
These techniques were further used for checking satisfiability of CTL \cite{Fast-flexible-CTL-2016}.
In this paper, we extend this approach to ATL.
We start with showing a Boolean encoding of the ATL models.

\subsection{Boolean Encoding of ATL Models}
\label{subsection:model-coding}

First, we make some assumptions about MAS. Assume that we are
given a set of agents ${\cal A}=\{1,\dots,n\}$, where each agent
$i \in {\cal A}$ has a fixed set of the local states
$L_i=\{l^1_i,\dots,l^{n_i}_{i}\}$ and a fixed initial local state
$\iota_i \in L_i$.
%
Since agent $i$ can be in one of its $n_i$ local states,
and a local transition function $T_i$ is restricted such that it does not involve
actions of the other agents,
we can assume, without a loss of generality, that agent $i$ has exactly $n_i$
possible actions, i.e., from each local state it can potentially move to each
of its local states.
%
So, assume that the set of local actions for agent $i$ is
$Act_i=\{a^1_i,\dots,a^{n_i}_i\}$ and an action $a^j_i$ can move
the agent $i$ from any local state to local state $l^j_i$.
Moreover, we assume that each local protocol $P_i$ satisfies that
at least one action is available at each local state.
Consequently, the local transition function $T_i$ for agent $i$ is
defined as follows: $T_i(l^k_i,a^j_i)=l^j_i$ if $a^j_i \in
P_i(l^k_i)$, for any $l^k_i \in L_i$ and $1 \leq j \leq n_i$.

Next, we represent every single agent $i$ with a given $\AG_i
=(L_i, \iota_i, Act_i, P_i, T_i)$ by means of a bit vector. In
fact, under the condition that the number of the local states is
fixed, the initial state is selected, and the rules for defining
the local actions and a local transition function are given, we
have to encode a local protocol $P_i$. It can be defined by a
Boolean table $lp_i$ of $|L_i| \times |Act_i|$ entries, where $0$
at position $(l^k_i,a^j_i)$ means that the local action $a^j_i$ is
not available at the local state $l^k_i$, and $1$ stands for the
availability. This table can be represented by a bit vector
$tb_i = (lp_i[1],\dots,lp_i[n_i])$\footnote{In what follows, we assume that
a sequence of bit vectors is identified with the bit vector composed of its elements.},
where $lp_i[j]$ stands for the $j$-th row of the table $lp_i$,
encoding which local actions are available at which local states.

Since the model $M=(\States,\iota,T,V)$ induced by a MAS
is a product of $\AG_i$ for $i \in {\cal A}$,
the bit vector
$(tb_1,\ldots,tb_n)$
determines the synchronous product of the local transition functions of
the agents and thus the global transition function $T$ of $M$.

Finally, we need to define a valuation of the propositional variables.
Given a set $\PV$, a Boolean table of size $|\States| \times |\PV|$
saves which propositional variables are true in which global states.
Then, let $vb=(vb_1,\dots,vb_{k})$ be a bit vector, where
$k=|\States|\cdot|\PV|$,
controlling which propositional variables hold in each global
state.

In this way, every model can be represented with a bit vector.
For a fixed number $|\PV|$ of the propositional variables,
a fixed number $n$ of agents,
a fixed number $n_i$ of the local states of agent $i$,
for every $i=1,\dots,n$,
the bit vector
$v_M = (tb_1,\ldots,tb_n,vb)$ encodes some model induced by MAS
without an initial state fixed. Therefore, $v_M$ actually encodes
a family of models which differ only in the initial state.

\subsection{Predicate \textit{Model}}
\label{subsection:predicate-model}

From now on, we consider models $M$ defined over the fixed number $|\PV|$ of the propositional variables
and a fixed number $n$ of agents with fixed numbers $|L_1|,\dots,|L_n|$ of local states.
Thus, we consider models that can be represented by a bit vector $v_M$ consisting of exactly
$n_M=|L_1|^2+$ $\dots$ $+|L_n|^2+|L_1|\cdot$ $\dots$ $\cdot|L_n|\cdot|\PV|$ bits.
In the rest of the work we will use the following notation:
$$\varM=(TB_1,\dots,TB_n,VB)$$
to denote a vector of Boolean variables, where for $i=1,\dots,n$, $TB_i$ is a vector of $|L_i|^2$
variables and $VB$ is a vector of $|L_1|\cdot\ldots\cdot|L_n|\cdot|\PV|$ variables.

For an ATL formula $\phi$ defined over propositional variables of $\PV$
and over agents of ${\cal A}$,
for each global state $g \in \States$ the following predicate is defined:
$Model_{g,\phi}(\varM)$.
%
%
For the bit vector $v_M$ encoding a model $M$
we define:
$Model_{g,\phi}(v_M)=1$ if and only if $M,g \models \phi$.
Unfortunately, it turns out that this predicate is not
monotonic, i.e. there is an ATL formula $\phi$ and a global state
$g$ for which the predicate $Model_{g,\phi}(\varM)$ is not monotonic w.r.t. $\varM$.
%

\begin{theorem}
\label{theorem:non-monotonicity-of-ATL}
The predicate $Model_{g,\phi}(\varM)$ is neither positive nor negative
monotonic w.r.t $\varM$.
\end{theorem}

\begin{proof}
Since ATL subsumes CTL, the thesis follows from the similar result for CTL \cite{Fast-flexible-CTL-2016}.
\end{proof}

However, in some special cases, as we show below, the predicate $Model_{g,\phi}(\varM)$ can be monotonic.

\begin{theorem}
\label{theorem:VB-positive-monotonicity-of-p-and-AX-AG-AU}

The predicate $Model_{g,\phi}(\varM)$ is positive monotonic w.r.t.
$VB$ if $\phi \in \{p, p\wedge q,\coop{{\GA}}\Next p, \coop{{\GA}}\Always p, \coop{{\GA}} p\Until q\}$,
where $p, q \in \PV$, $\GA \subseteq {\Ag}$.
\end{theorem}
\begin{proof}
Let $\phi \!\in \! \{p, p\wedge q, \coop{{\GA}}\Next p, \coop{{\GA}}\Always p, \coop{{\GA}} p \Until q\}$,
where $p, q \in \PV$ and let $v_M$ be a bit vector such that $Model_{g,\phi}(v_M)=1$.
This means that $M,g \models \phi$, for $M$ encoded by $v_M$, where $g$ is an initial state of $M$.

Now, let $v_{M'}$ be a vector which differs from $v_M$ only in one value $vb_j$,
for $j=1,\dots,k$, which is $0$ in $v_M$ and $1$ in $v_{M'}$.
The model $M'$, encoded by $v_{M'}$, has the same states, transitions, and state properties
as $M$, except for one state property which holds in $M'$ but not in $M$, i.e.,
one propositional variable holds true in some state $t$ in $M'$ but does not hold in $t$ in $M$.
Thus, if $\phi \in \{p,p\wedge q\}$ and $M,g \models \phi$, then $M',g \models \phi$ as well.

Consider the case of $\phi \in \{\coop{{\GA}}\Next p, \coop{{\GA}}\Always p, \coop{{\GA}} p \Until q\}$.
Since $M$ is a model of $\phi$, then there is a strategy $\strat_{\GA}\in \CollStra$ such that
for each path $\seq \in \outcome_\model(\stan,\strat_{\GA})$, $\seq \models \psi$
for $\psi \in \{\Next p, \Always p, p \Until q\}$.
Clearly, there is the same strategy $\strat_{\GA}\in \CollStra$ in $M'$.
Consider a path $\seq' \in \outcome_{\model'}(\stan,\strat_{\GA})$.
This path differs from the corresponding path $\seq \in \outcome_\model(\stan,\strat_{\GA})$
such that it may contain more states where $p$ or $q$ holds.
Therefore, $\seq' \models \psi$ for  $\psi \in \{\Next p, \Always p, p \Until q\}$.
So, we have $Model_{g,\phi}(v_{M'})=1$.
\end{proof}

\begin{theorem}
\label{theorem:VB-negative-monotonicity-of-negation}

The predicate $Model_{g, \neg p}(\varM)$ is negative monotonic
w.r.t. $VB$, for $p \in \PV$.
\end{theorem}
\begin{proof}
Let $v_M$ be a bit vector such that $Model_{g,\neg p}(v_M)=1$.
This means that $M,g \models \neg p$ for $M$, encoded by $v_M$, with the initial state $g$.
Now, let $v_{M'}$ be a bit vector which differs from $v_M$ only in one value $vb_j$,
for $j=1,\dots,k$, which is $1$ in $v_M$ and $0$ in $v_{M'}$.
The model $M'$, encoded by $v_{M'}$, has the same states, transitions, and state properties as $M$,
except for one state property which does not hold in $M'$ but holds in $M$,
i.e., one propositional variables is false in some state $t$ in $M'$
but is true in $t$ in $M$.
Thus if $M,g \models \neg p$, then $M',g \models \neg p$ and finally $Model_{g,\neg p}(v_{M'})=1$.
\end{proof}

\begin{theorem}
\label{theorem:TB-positive-negative-monotonicity-of-p-and-negation}

The predicate $Model_{g,\phi}(\varM)$ is both positive
and negative monotonic w.r.t. $TB_i$ for each $i \in {\Ag}$
if $\phi \in \{ p, \neg p, p \wedge q \}$, where $p,q \in \PV$.
\end{theorem}
\begin{proof}
Notice that adding or removing transitions (both local or global)
does not alter the truthfulness of the formula $\phi \in \{ p,
\neg p, p \wedge q \}$ as long as $p$ and $q$ are propositional
variables. Therefore, $Model_{g,\phi}(\varM)$ for $\phi \in \{p,\neg p, p\wedge q \}$ is both positive and negative
monotonic w.r.t. $TB_i$ for each $i \in {\Ag}$.
\end{proof}


\begin{theorem}
\label{theorem:TBi-positive-monotonicity-of-AX-AG-AU}

The predicate $Model_{g,\phi}(\varM)$ is positive monotonic w.r.t.
$TB_i$ for $i \in {\GA}$ if $\phi \in \{\coop{{\GA}}\Next p,
\coop{{\GA}}\Always p$, $\coop{{\GA}} p \Until q\}$, where $p, q
\in \PV$, $\GA \subseteq {\Ag}$.
\end{theorem}
\begin{proof}
Let $\phi \in \{\coop{{\GA}}\Next p, \coop{{\GA}}\Always p, \coop{{\GA}} p \Until q\}$, where $p, q \in \PV$
and let $v_M$ be a bit vector such that $Model_{g,\phi}(v_M)=1$.
This means that $M,g \models \phi$ for $M$, encoded by $v_M$, with the initial state $g$.
Now, let $v_{M'}$ be a vector which differs from $v_M$ only in one
value $tb^j_i$, for some $i \in {\GA}$ and
$j \in \{1,\dots,(n_i)^2\}$,  which is $0$ in $v_M$ and $1$ in $v_{M'}$.
The model $M'$, encoded by $v_{M'}$, has the same states, state
properties, and local transitions of the agents, except for one
local transition of one agent from ${\GA}$ that is enabled in $M'$
but not in $M$.

If $Model_{g,\phi}(v_M)=1$, then there is a strategy $\strat_{\GA}\in \CollStra$ such that for
each path $\seq\in \outcome_\model(\stan,\strat_{\GA})$, $\seq \models \psi$
for $\psi \in  \{\Next p, \Always p, p \Until q\}$.

Observe that adding one local transition to one agent of ${\GA}$ results
in more strategies of the agents of ${\GA}$, but at the same time the existing strategies are still in place.
Therefore, the strategy $\strat_{\GA} \in \CollStra$ is in $M'$ as well.
Therefore, $Model_{g,\phi}(v_{M'})=1$.
\end{proof}

\begin{theorem}
\label{theorem:TBi-negative-monotonicity-of-AX-AG-AU}

The predicate $Model_{g,\phi}(\varM)$ is negative monotonic
w.r.t.
$TB_i$ for $i \in {\mathcal A} \setminus {\GA}$ if $\phi \in \{\coop{{\GA}}\Next p$, $\coop{{\GA}}\Always p, \coop{{\GA}} p \Until q\}$,
where $p, q \in \PV$, $\GA \subseteq {\Ag}$.
\end{theorem}
\begin{proof}
Let $\phi \in \{\coop{{\GA}}\Next p,\coop{{\GA}}\Always p, \coop{{\GA}} p \Until q\}$, where $p, q \in \PV$, $\GA \subseteq {\Ag}$,
and $v_M$ be a bit vector s.t. $Model_{g,\phi}(v_M)=1$.
This means that $M,g \models \phi$ for $M$, encoded by $v_M$, with the initial state $g$.

Now, let $v_{M'}$ be a bit vector which differs from $v_M$ only in one value $tb^j_i$,
for some $i\in {\cal A} \setminus {\GA}$ and $j=1,\dots,(n_i)^2$, which is $1$ in $v_M$ and $0$ in $v_{M'}$.
The model $M'$, encoded by $v_{M'}$, has the same states, state properties, and local transitions of the agents,
except for one local transition of one agent of ${\mathcal A} \setminus {\GA}$ that is enabled in $M$
but not in $M'$.
Observe that deleting one local transition of some agent of ${\mathcal A} \setminus {\GA}$
results in the same number of strategies of the agents of ${\GA}$, but
for each strategy the number of paths in its outcome may be lower.
The protocol function ensures that at least one action and thereby at least one transition must remain
(not all can be deleted).
Thus, for any strategy of the agents of ${\GA}$, the number of transitions consistent with this strategy
cannot be reduced to zero.
If $Model_{g,\phi}(v_M)=1$, then there is a strategy $\strat_{\GA}\in \CollStra$ such that for each path
$\seq \in \outcome_\model(g,\strat_{\GA})$,
$\seq \models \psi$ for $\psi \in \{\Next p,\Always p, p \Until q\}$.
Therefore, the strategy $\strat_{\GA}\in \CollStra$ is in $M'$.
Since $\emptyset \neq \outcome_{\model'}(g,\strat_{\GA}) \subseteq \outcome_\model(g,\strat_{\GA})$,
we have $Model_{g,\phi}(v_{M'}) = 1$.
\end{proof}

\subsection{Function \textit{solve}}

In order to compute the value of the predicate $Model_{g,\phi}(v_M)$ for a given $M$, we define a new function, called $solve_{\phi}(V_m)$.
This function returns a set of states of $M$ such that $g \in solve_{\phi}(v_M)$ iff $Model_{g,\phi}(v_M)=1$, i.e., $M,g \models \phi$.
The monotonicity properties also apply to the function $solve_{\phi}$, as every state returned by this function can be viewed as an initial
state of the model $M$. Thus, the theorem below follows directly from Theorems
\ref{theorem:VB-positive-monotonicity-of-p-and-AX-AG-AU} --
\ref{theorem:TBi-negative-monotonicity-of-AX-AG-AU}.

\begin{theorem}
\label{theorem:monotonicity-of-solve}
The function $solve_{\phi}(\varM)$ is
\begin{itemize}
\item positive monotonic w.r.t. $VB$ for $\phi \in \{p,p\wedge q,$
$\coop{{\GA}}\Next p$, $\coop{{\GA}}\Always p,$ $\coop{{\GA}} p
\Until q\}$,

\item negative monotonic w.r.t. $VB$ for $\phi =\neg p$,

\item positive and negative monotonic w.r.t. $TB_i$ for $i \in
{\Ag}$ if $\phi \in \{ p,\neg p, p\wedge q \}$,

\item positive monotonic w.r.t. $TB_i$ for each $i \in {\GA}$ if $\phi
\in \{\coop{{\GA}}\Next p$, $\coop{{\GA}}\Always p,$ $\coop{{\GA}}
p \Until q\}$,

\item negative monotonic w.r.t. $TB_i$ for $i \in {\Ag} \setminus
{\GA}$ if $\phi \in \{ \coop{{\GA}}\Next p,$ \\ $\coop{{\GA}}\Always p,
\coop{{\GA}} p \Until q\}$,

where $p,q \in \PV$, $\GA \subseteq \Ag$.
\end{itemize}
\end{theorem}
%

Moreover, to compute $solve_{\phi}(\varM)$ for each ATL formula $\phi$, a new evaluation function $solve_{op}(Y_1,\varM)$ is defined
for an unary operator $op$ and $solve_{op}(Y_1,Y_2,\varM)$ for a binary operator $op$,
and $Y_1, Y_2 \subseteq \States$.
This function evaluates the operator $op$ on sets of states $Y_1$, $Y_2$ instead of the formulae holding in these states.
If $\phi = p \in \PV$, then for a given model $M$, $solve_{p}(v_M)$ returns the set of states of $M$ in which $p$ holds.
Otherwise,  $solve_{\phi}(v_M)$ takes the top-most operator $op$ of $\phi$ and solves its argument(s) recursively using the function $solve_{op}$
and applying $solve_{op}(Y_1,v_M)$ ($solve_{op}(Y_1,Y_2,v_M)$) to the returned set(s) of states.

Now, Theorem \ref{theorem:monotonicity-of-solve} can be rewritten by replacing propositional variables $p$ and $q$ by sets of states
satisfying these variables.

\begin{theorem}
\label{theorem:monotonicity-of-solve-op}
The function $solve_{op}(Y_1,\varM)$ for an unary operator $op$ and $solve_{op}(Y_1,Y_2,\varM)$ for a binary operator $op$ is
\begin{itemize}
\item positive monotonic w.r.t. $VB$ for $op \!\in\! \{\wedge,\! \coop{{\GA}}\Next,\! \coop{{\GA}}\Always$,\! $\coop{{\GA}}\Until \}$,

\item negative monotonic w.r.t. $VB$ for $op =\neg$,

\item positive and negative monotonic w.r.t. $TB_i$ for $i \in
{\Ag}$ and $op \in \{\neg, \wedge \}$,

\item positive monotonic w.r.t. $TB_i$ for $i\! \in\! {\GA}$ and $op\! \in\! \{\coop{{\GA}}\Next, \coop{{\GA}}\Always,$ $\coop{{\GA}}\Until \}$,

\item negative monotonic w.r.t. $TB_i$ for $i \in {\Ag} \setminus {\GA}$ and $op \in \{ \coop{{\GA}}\Next,$ $\coop{{\GA}}\Always,\coop{{\GA}}\Until \}$.
\end{itemize}
\end{theorem}
%
%
%
\noindent
To compute $solve_{\phi}(v_M)$, $solve_{op}(Y_1,v_M)$, and $solve_{op}(Y_1,Y_2,v_M)$ the model checking
algorithms
described in \cite{MCMAS-Lomuscio2017} are applied.

\section{Approximating ATL Models}
\label{section:ATL-formula-approximation}

In this section we show how to approximate models for ATL in order to solve the satisfiability
problem using SAT modulo monotonic theories.
First, the construction of over and under approximations of a model are given.
Then, the approximation algorithm is defined together with the proofs of its properties.

\subsection{Construction of $M_{over}$ and $M_{under}$ }
\label{subsection:structures-over-under}

Given a set of agents $\Ag=\{1,\dots,n\}$, we fix for each $i \in \Ag$ a set of local states $L_i$,
an initial state $\iota_i$, and a set of local actions 
defined like in Def. \ref{def:MAS}.
Next we define a function, called a \textit{partial protocol}:
%
$$CP_i: L_i \times Act_i \rightarrow \{0,1,undef\}.$$
By a \textit{partial MAS}, denoted $MAS_{CP}$, we mean a MAS in which each agent is associated
with a partial protocol rather than with a protocol.
Then, a model induced by $MAS_{CP}$ together with a \textit{partial valuation of the propositional variables}
$$CV:\States \times \PV \rightarrow \{ 0,1,undef \}$$
is called a \textit{partial model}, denoted by $M_{par}$.
Both a partial protocol and a partial valuation can be extended to total
functions.
The intention behind these definitions is to give requirements on the models. 

For each partial model, total models $M^{\Ag}_{under}$ and $M^{\GA}_{over}$, for ${\GA} \subseteq {\mathcal A}$, are constructed.
%
%
First, for every agent $i \in \Ag$ we define:
\textit{a necessary local protocol}  $\underline{P_i}: L_i \to 2^{Act_i}$
and \textit{a possible local protocol}  $\overline{P_i}: L_i \to 2^{Act_i}$, where:\\
 (1) if
 $CP_i(l_i,a_i)=1$ then $a_i \in \underline{P_i}(l_i)$ and $a_i \in \overline{P_i}(l_i)$,\\
 (2) if
 $CP_i(l_i,a_i)=0$ then $a_i \not \in \underline{P_i}(l_i)$ and $a_i \not \in \overline{P_i}(l_i)$,\\
 (3) if
 $CP_i(l_i,a_i)=undef$ then $a_i \not\in \underline{P_i}(l_i)$ and $a_i \in \overline{P_i}(l_i)$.\\
Notice that the possible local protocol is an extension of the
necessary local protocol, i.e., the following condition holds:
for every local state $l_{i}$,
$\underline{P_{i}}(l_{i}) \subseteq \overline{P_{i}}(l_{i})$.
In a similar way, total valuations of the propositional variables are defined:
%
%
%
\textit{a necessary valuation}
 $\underline{V}: \States \to 2^{\PV}$
and
\textit{a possible valuation}
 $\overline{V}: \States \to 2^{\PV}$ such that:\\
 (1) if
 $CV(g,p)=1$ then $p \in \underline{V}(g)$ and $p \in \overline{V}(g)$,\\
 (2) if
 $CV(g,p)=0$ then $p \not\in \underline{V}(g)$ and $p \not\in \overline{V}(g)$,\\
 (3) if
 $CV(g,p)=undef$ then $p \not\in \underline{V}(g)$ and $p \in \overline{V}(g)$.

\noindent
Observe that for every global state $g \in \States$ we have
$\underline{V}(g) \subseteq \overline{V}(g)$.

The model $M^{\Ag}_{under}$ is defined as in Def. \ref{def:IS} of all agents $i \in \Ag$ with
${\AG_i}=(L_{i},\iota_i,Act_{i},\underline{P_i},T_{i})$
and for the valuation of the propositional variables
$\underline{V}$.
The model $M_{over}^{\GA}$ is defined as in Def. \ref{def:IS} of agents $i \in
{\GA}$ with ${\AG_i}=(L_{i}, \iota_i,
Act_{i},\overline{P_i},T_{i})$, agents $j \in {\mathcal A}
\setminus {\GA}$ with ${\AG_j}=(L_{j}, \iota_j, Act_{j},\underline{P_j},T_{j})$, and for the valuation of the
propositional variables $\overline{V}$.

\subsection{Algorithm \textit{$\solveApprox$}}

We say that the model $M=(\States,\iota,T,V)$ induced by agents
$\Ag=\{1,\dots,n\}$ with $AG_i=(L_i,\iota_i,Act_i,P_i,T_i)$ for $i \in \Ag$ and
the propositional variables $\PV$ is \textit{compatible} with a partial model $M_{par}$
induced by the same sets of agents and propositional variables, and determined by
the given partial protocols $CP_i$ for $i \in \Ag$, and a partial valuation $CV$
if $P_i$ is consistent with $CP_i$ for every $i \in \Ag$, and $V$ satisfies all conditions determined by $CV$.
Formally:

\begin{itemize}

\item (1) if $CP_i(l_i,a_i)=1$ then $a_i \in P_i(l_i)$,\\
      (2) if $CP_i(l_i,a_i)=0$ then $a_i \not \in P_i(l_i)$,

\item (1) if $CV(g,p)=1$ then $p \in V(g)$,\\
      (2) if $CV(g,p)=0$ then $p \not\in V(g)$.

\end{itemize}
\noindent Observe that $M^{\Ag}_{under}$ and $M^{\GA}_{over}$ are
compatible with $M_{par}$.
What is more, for any model $M$ compatible with $M_{par}$ we have: 
$$\forall i \in \Ag\;\;\; \underline{P_{i}}(l_{i})\subseteq P_i(l_i) \subseteq \overline{P_{i}}(l_{i}),\;\;\;
\underline{V}(g) \subseteq V(g) \subseteq \overline{V}(g). $$

\begin{theorem}
\label{theorem:variable-inclusion} Let $v_M$,
$v_{M^{\Ag}_{under}},v_{M^{\GA}_{over}}$, for  some ${\GA}
\subseteq \Ag$, be bit vectors encoding models $M$,
$M^{\Ag}_{under}$, and $M^{\GA}_{over}$, respectively, then we
have:

\begin{itemize}
\item
$v_{M^{\Ag}_{under}}[tb_i[j_i]] \leq v_{M}[tb_i[j_i]] \leq v_{M^{\GA}_{over}}[tb_i[j_i]]$\\
      for all $i \in {\GA}$ and for all $1 \leq j_i \leq n_i$,
            for short

            $v_{M^{\Ag}_{under}}[TB_i] \leq v_{M}[TB_i] \leq v_{M^{\GA}_{over}}[TB_i]$ for $i \in {\GA}$, and

\item $v_{M^{\Ag}_{under}}[vb_j] \leq v_{M}[vb_j] \leq
v_{M^{\GA}_{over}}[vb_j]$ for all $1 \leq j \leq k$,

      for short $v_{M^{\Ag}_{under}}[VB] \leq v_{M}[VB] \leq v_{M^{\GA}_{over}}[VB]$.
\end{itemize}
\end{theorem}
\begin{proof}
Follows
from the definitions of $M^{\Ag}_{under}$ and
$M^{\GA}_{over}$.
\end{proof}
This means that each transition in $M^{\Ag}_{under}$ is also a
transition in $M$, and each transition in $M$ is a transition in
$M^{\GA}_{over}$. Similarly for the propositional variables, if
some propositional variable holds true at state $g$ of
$M^{\Ag}_{under}$, then it also holds true at the same state of
$M$, and if some propositional variable holds true at $g$ of $M$,
then it also holds true at the same state of $M^{\GA}_{over}$.

Now, a new function $\solveApprox_{M_{par}}(\phi,\varM[1],\varM[2])$ over two separate assignments of transitions and states
$\varM[1]=(TB^1_1,\dots, TB^1_n, VB^1)$ and $\varM[2]=(TB^2_1,\dots, TB^2_n, VB^2)$ is defined.
For a given partial model $M_{par}$ and two models $M_1$ and $M_2$ compatible with $M_{par}$,
the output of the function is determined by the following algorithm.

\vspace{0,2cm} \textbf{Algorithm} $\solveApprox_{M_{par}}(\phi, \bitM[1], \bitM[2])$

1: \textbf{if} $\phi \in \PV$ \textbf{then}

2: \hspace{0,5cm} \textbf{return} $\{ g \in \States :
M_1,g \models \phi \}$

3: \textbf{else if} $\phi=op(\psi)$ \textbf{then}

4: \hspace{0,5cm} \textbf{if} $op$ is $\neg$ \textbf{then}
\hspace{0,5cm}\hspace{0,5cm} // negative monotonic

5: \hspace{0,5cm}\hspace{0,5cm} $Y := \solveApprox_{M_{par}}(\psi,
 \bitM[2], \bitM[1])$

6: \hspace{0,5cm}\hspace{0,5cm} \textbf{return} $solve_{op}(Y,
\bitM[2])$

7: \hspace{0,5cm} \textbf{else} \hspace{0,5cm} \hspace{0,5cm}
\hspace{0,9cm}// $op \in \{\coop{\GA}\Next, \coop{\GA}\Always\}$

8: \hspace{0,5cm}\hspace{0,5cm}$Y:=\solveApprox_{M_{par}}(\psi,
\bitM[1], \bitM[2])$

9: \hspace{0,5cm}\hspace{0,5cm}\textbf{if} $\bitM[1]=\MAunder$
\textbf{then}

10: \hspace{0,5cm}\hspace{0,5cm}\hspace{0,5cm} \textbf{return}
$solve_{op}(Y, \bitM[1])$

11: \hspace{0,2cm}\hspace{0,5cm} \textbf{else} \textbf{return}
$solve_{op}(Y, \MGover)$

12: \textbf{else if}  $\phi \in \{ \coop{\GA}\psi_1 \Until \psi_2,
\psi_1 \wedge \psi_2 \}$

13: \hspace{0,5cm} $Y_1:= \solveApprox_{M_{par}}(\psi_1, \bitM[1],\bitM[2])$

14: \hspace{0,5cm} $Y_2:= \solveApprox_{M_{par}}(\psi_2,\bitM[1],\bitM[2])$

15: \hspace{0,5cm} \textbf{if} $\phi$ is $\coop{\GA}\psi_1 \Until
\psi_2$ \textbf{then}

16: \hspace{0,5cm}\hspace{0,5cm} \textbf{if} $\bitM[1]=\MAunder$
\textbf{then}

17: \hspace{0,5cm}\hspace{0,5cm}\hspace{0,5cm} \textbf{return}
$solve_{op}(Y_1,Y_2, \bitM[1])$

18: \hspace{0,5cm}\hspace{0,5cm} \textbf{else} \textbf{return}
$solve_{op}(Y_1,Y_2, \MGover)$

19: \hspace{0,5cm} \textbf{else} \hspace{3,5cm} // $op=\wedge$

20: \hspace{0,5cm}\hspace{0,5cm} \textbf{return}
$solve_{op}(Y_1,Y_2, \bitM[1])$



\begin{theorem}
\label{theorem:solveApprox-is-monotonic}
The function $\solveApprox_{M_{par}}(\phi, \varM[1], \varM[2])$ is
\begin{itemize}
\item positive monotonic w.r.t. $TB^{1}_i$ for $i \in {\Ag}$ and $VB^{1}$ and
negative monotonic w.r.t. $TB^{2}_i$ for $i \in {\Ag}$ and $VB^{2}$ for
$\phi \in \{p, \neg p, p \wedge q \}$, and

\item positive monotonic w.r.t. $TB^{1}_i$ for $i \in {\GA}$ and $VB^{1}$ and
negative monotonic w.r.t. $TB^{2}_i$ for $i \in {\GA}$ and $VB^{2}$ for
$\phi \in \{ \coop{{\GA}}\Next p$, $\coop{{\GA}}\Always p,$ $\coop{{\GA}} p\Until q\} \}$.
\end{itemize}
%
%
\end{theorem}


\begin{proof}
By a structural induction on a formula $\phi$.
Let $M_{11}$, $M_{12}$, $M_{21}$, $M_{22}$ be models compatible with $M_{par}$ such that
$\bitM[11][VB] \leq \bitM[12][VB]$, $\bitM[11][TB_i] \leq \bitM[12][TB_i]$ for $i \in \Ag$, and
$\bitM[22][VB] \leq \bitM[21][VB]$, $\bitM[22][TB_i] \leq \bitM[21][TB_i]$ for $i \in \Ag$.

\noindent
\textbf{The base case}. If $\phi=p \in \PV$, then from the definition of the algorithm,
$\solveApprox_{M_{par}}(p,\bitM[11],\bitM[21])$ returns the set of the states satisfying $p$ in $M_{11}$ and
$\solveApprox_{M_{par}}(p,\bitM[12],\bitM[21])$ returns the set of the states satisfying $p$ in $M_{12}$.
Since $\bitM[11][VB] \leq \bitM[12][VB]$ then $\solveApprox_{M_{par}}(p,\bitM[11],\bitM[21])$ $\subseteq$
$\solveApprox_{M_{par}}(p,\bitM[12],\bitM[21])$ and \\
$\solveApprox_{M_{par}}(\phi, \varM[1], \varM[2])$ is positive monotonic w.r.t. $VB^1$.
Observe that the output of $\solveApprox_{M_{par}}(\phi, \varM[1], \varM[2])$ depends only on values of variables $VB^1$,
thus the function is also positive monotonic w.r.t. $TB^1_i$ for $i \in \Ag$ and negative monotonic w.r.t. $VB^2$ and $TB^2_i$ for $i \in \Ag$.

\noindent \textbf{The induction step}.
We show the proof for the unary operators $\neg, \coop{\GA}\Next, \coop{\GA}\Always$.
The proofs for the binary operators $Until$ and $\wedge$ are similar.

\noindent Induction assumption (IA): the thesis holds for a formula $\psi$.
Induction hypothesis (IH): the thesis holds for $\phi = op \; \psi$.

\noindent $\bullet$ If $\phi = \neg \psi$.

\noindent
If $Y\!=\!\solveApprox_{M_{par}}(\psi,\bitM[21],\bitM[11])$ and $Y'=\solveApprox_{M_{par}}(\psi,\bitM[21],\bitM[12])$, then $Y' \subseteq Y$ since
$\solveApprox_{M_{par}}$ is negative monotonic w.r.t. $VB^{2}$ and $TB^{2}_i$ for $i \in {\Ag}$, from IA.
Next, $solve_{\neg}(Y,\bitM[21])$ $\subseteq$ $solve_{\neg}(Y',\bitM[21])$ since $solve_{\neg}$ returns the compliment of $Y$ and $Y'$, respectively.
Thus, $\solveApprox_{M_{par}}(\phi,\bitM[11],\bitM[21])$$\subseteq$$\solveApprox_{M_{par}}(\phi,\bitM[12],\bitM[21])$ and the function is
positive monotonic w.r.t. $VB^1$ and $TB^1_i$ for $i \in \Ag$.

\noindent
If $Y\!=\!\solveApprox_{M_{par}}(\psi,\bitM[21],\bitM[11])$ and $Y'=\solveApprox_{M_{par}}(\psi,\bitM[22],\bitM[11])$, then $Y' \subseteq Y$ since
$\solveApprox_{M_{par}}$ is positive monotonic w.r.t. $VB^{1}$ and $TB^{1}_i$ for $i \in {\Ag}$, from IA.
Next, $solve_{\neg}(Y,\bitM[21])$ $\subseteq$ $solve_{\neg}(Y',\bitM[21])$ since $solve_{\neg}$ returns the compliment of $Y$ and $Y'$, respectively, and
$solve_{\neg}(Y',\bitM[21])$ $\subseteq$ $solve_{\neg}(Y',\bitM[22])$ since $solve_{\neg}$ is negative monotonic w.r.t. $VB$ and $TB_i$ for $i \in \Ag$.
Thus, $solve_{\neg}(Y,\bitM[21])$ $\subseteq$ $solve_{\neg}(Y',\bitM[22])$, i.e.,
$\solveApprox_{M_{par}}(\phi,$ $\bitM[11],$ $\bitM[21])$ $\subseteq$
$\solveApprox_{M_{par}}(\phi,$ $\bitM[11],$ $\bitM[22])$ and the function is
negative monotonic w.r.t. $VB^2$ and $TB^2_i$ for $i \in \Ag$.


\noindent $\bullet$ If $\phi = op \; \psi$ with $op\in \coop{\GA}\Next, \coop{\GA}\Always$.

Now, let $M_{11}$, $M_{12}$, $M_{21}$, $M_{22}$ be models such that
$\bitM[11][TB_i] \leq \bitM[12][TB_i]$ and $\bitM[22][TB_i] \leq \bitM[21][TB_i]$ for $i \in \GA \subseteq \Ag$.
The other restrictions remain the same.
%
If $Y\!=\!\solveApprox_{M_{par}}(\psi,\bitM[11],\bitM[21])$ and $Y'\!=\!\solveApprox_{M_{par}}(\psi,\bitM[12],\bitM[21])$, then $Y \subseteq Y'$ since
$\solveApprox_{M_{par}}$ is positive monotonic w.r.t. $VB^{1}$ and $TB^{1}_i$ for $i \in {\GA}$, from IA.
%
Next, $\solveApprox_{M_{par}}(\phi,\bitM[11],\bitM[21])$ is $solve_{op}(Y,$ $\bitM[1])$ where $M_1$ is $M^{\Ag}_{under}$ or $M^{\GA}_{over}$.
Similarly, $\solveApprox_{M_{par}}(\phi,\bitM[12],$ $\bitM[21])$ is $solve_{op}(Y',$ $\bitM[2])$ where $M_{2}$ is $M^{\Ag}_{under}$ or $M^{\GA}_{over}$.
In all the cases, if $\bitM[11](VB)\leq \bitM[12](VB)$ and $\bitM[11](TB_i)\leq \bitM[12](TB_i)$ for $i \in \GA$ then
$\bitM[1](VB)\leq \bitM[2](VB)$ and $\bitM[1](TB_i)\leq \bitM[2](TB_i)$ for $i \in \GA$.
Now observe that
if $Y \subseteq Y'$ then $solve_{op}(Y,\bitM[1])$ $\subseteq$ $solve_{op}(Y',\bitM[1])$.\\
Next, $solve_{op}(Y',\bitM[1]) \subseteq solve_{op}(Y',\bitM[2])$ since $solve_{op}$ is positive monotonic w.r.t. $VB$ and $TB_i$ for $i \in \GA$.
Finally, $solve_{op}(Y,\bitM[1])$ $\subseteq$ $solve_{op}(Y',\bitM[2])$, i.e.,
$\solveApprox_{M_{par}}(\phi,\bitM[11],\bitM[21]) \subseteq \solveApprox_{M_{par}}(\phi,$ $\bitM[12],$ $\bitM[21])$, and the function is
positive monotonic w.r.t. $VB^1$ and $TB^1_i$ for $i \in \GA$.

\noindent
If $Y\!=\!\solveApprox_{M_{par}}(\psi,\bitM[11],\bitM[21])$ and $Y'\!=\!\solveApprox_{M_{par}}(\psi,\bitM[11],\bitM[22])$, then $Y \subseteq Y'$ since
$\solveApprox_{M_{par}}$ is negative monotonic w.r.t. $VB^{2}$ and $TB^{2}_i$ for $i \in {\GA}$, from IA.
%
%
The rest of the proof proceeds similarly like in the case above.
Finally, $\solveApprox_{M_{par}}(\phi,\bitM[11],\bitM[21]) \subseteq \solveApprox_{M_{par}}(\phi,\bitM[11],\bitM[22])$, and the function is
negative monotonic w.r.t. $VB^2$ and $TB^2_i$ for $i \in \GA$.
\end{proof}

The algorithm $\solveApprox_{M_{par}}(\phi,\varM[1],\varM[2])$, for a model $M$ compatible with $M_{par}$,
computes over and under-approximation of $solve_{\phi}( \bitM )$.
More precisely, $\solveApprox_{M_{par}}\!(\phi,\MAover, \MAunder)$ returns a set of
states represented by a bit vector, which is an over-approximation
of $solve_{\phi}( \bitM )$ for a model $M$ compatible with $M_{par}$.
This means that if $\iota \in solve_{\phi}(\bitM)$, then $\iota \in \solveApprox_{M_{par}}(\phi,\MAover, \MAunder)$.
Clearly, if $\iota \not\in \solveApprox_{M_{par}}(\phi,\MAover, \MAunder)$, then
there is no model $M$ extending $M_{par}$ such that $M,\iota \models \phi$.
Similarly, $\solveApprox_{M_{par}}(\phi,$ $\MAunder,$ $\MAover)$ computes an under-approximation of $solve_{\phi}(\bitM)$.
This means that if $\iota \!\in \!\solveApprox_{M_{par}}(\phi$, $\MAunder,\MAover)$
then $\iota \!\in\! solve_{\phi}(\bitM)$.


\begin{theorem}\label{theorem:solveApprox-in-solve}
Let $M_{par}$ be a partial model and $M$ be a model compatible
with $M_{par}$. Then, for any ATL formulae $\phi,\psi_1,\psi_2$
and $p \in \PV$, we have:
\begin{enumerate}
\item for $\phi \in \{p, \neg \psi_1, \psi_1 \wedge \psi_2\}$ and
each ${\GA} \subseteq {\mathcal A}$:
\begin{description}
\item[]
$\solveApprox_{M_{par}}(\phi,\MAunder,\MGover)
\subseteq solve_{\phi}(\bitM)$;

\item[] $solve_{\phi}(M) \subseteq
\solveApprox_{M_{par}}(\phi,\MGover, \MAunder)$;
\end{description}

\item for $\phi \in \{\coop{\GA}\Next \psi_1, \coop{\GA}\Always
\psi_1,\coop{\GA} \psi_1 \Until \psi_2 \}$:
\begin{description}
\item[] $\solveApprox_{M_{par}}(\phi, \MAunder,\MGover) \subseteq solve_{\phi}(\bitM)$;

\item[] $solve_{\phi}(M) \subseteq
\solveApprox_{M_{par}}(\phi,\MGover, \MAunder)$.
\end{description}
\end{enumerate}
\end{theorem}
\begin{proof}
By a structural induction on a formula. We prove that

\noindent {\bf (a)} $solve_{\phi}(\bitM) \subseteq
\solveApprox_{M_{par}}(\phi,\MGover, \MAunder)$ and

\noindent {\bf (b)} $\solveApprox_{M_{par}}(\phi, \MAunder, \MGover) \subseteq solve_{\phi}(\bitM)$.

\noindent \textbf{The base case}.
If $\phi = p \in \PV$, then $solve_{\phi}(\bitM)$ returns the set of states satisfying $p$ in $M$,
$\solveApprox_{M_{par}}(\phi,\MGover,\MAunder)$ returns the set of states
satisfying $p$ in $M_{over}^{\GA}$, and
$\solveApprox_{M_{par}}(\phi,\MAunder,$ $\MGover)$ returns the set of states
satisfying $p$ in $M_{under}^{\Ag}$.
The thesis holds by Theorem
\ref{theorem:variable-inclusion} which implies that
$v_{M_{under}^{\Ag}}[VB] \leq  v_{M}[VB] \leq v_{M_{over}^{\GA}}[VB]$,
i.e., the states satisfying $p$ in $M$ are included in the states satisfying $p$
in $M^{\GA}_{over}$ and the states satisfying $p$ in $M^{\Ag}_{under}$ are included in the
states satisfying $p$ in $M$.
Notice that this does not depend on the set of agents ${\GA}$ since
the models $M$, $M^{\GA}_{over}$, $M^{\Ag}_{under}$ have the same states
and the transitions do not affect the values of the propositional variables in the states.

\noindent \textbf{The induction step}.
We show the proof for {\bf (a)} and for
$\phi \in \{\neg \psi,$ $\coop{\GA}\Next \psi, \coop{\GA}\Always \psi \}$.
The rest of the proof proceeds similarly.

\noindent Induction assumption (IA):
the thesis holds for a formula $\psi$.
Induction hypothesis (IH): the thesis holds for $\phi = op \; \psi$.


\noindent $\bullet$ If $\phi = \neg \psi$,
then

\noindent $solve_{\phi}(\bitM) = solve_{\neg}(Y,\bitM)$ for $Y = solve_{\psi}(\bitM)$ and

\noindent
$\solveApprox_{M_{par}}(\phi,\MGover,\MAunder)$ $=$
$solve_{\neg}(Y',\MAunder)$
for $Y' =$\\
$\solveApprox_{M_{par}}(\psi,\MAunder,\MGover)$ and
any $\GA$.

Observe that $solve_{\neg}(Y,\bitM)$ returns the compliment of $Y$, i.e. $\States \setminus Y$.
Thus, $solve_{\neg}(Y,\bitM) \subseteq solve_{\neg}(Y',\bitM)$ since $Y' \subseteq Y$ from IA.
Next, $solve_{\neg}(Y',\bitM) \subseteq solve_{\neg}(Y',\MAunder)$ since function $solve_{\neg}(Y,\varM)$ is negative monotonic w.r.t. $VB$ and $TB_i$ for $i \in \Ag$
from Theorem \ref{theorem:monotonicity-of-solve-op} and $v_{M_{under}^{\Ag}}\![VB] \leq v_{M}[VB]$ and $v_{M_{under}^{\Ag}}\![TB_i] \leq v_{M}[TB_i]$ for $i \in \Ag$
from Theorem \ref{theorem:variable-inclusion}. Finally, $solve_{\neg}(Y,\bitM)$ $\subseteq$ $solve_{\neg}(Y',\MAunder)$ and thus
$solve_{\phi}(\bitM) \subseteq \solveApprox_{M_{par}}(\phi$,
$ \MGover,$ $\MAunder)$.

\noindent $\bullet$ If $\phi = op \; \psi$ with $op \in \{\coop{\GA}\Next,\coop{\GA}\Always\}$, then

\noindent $solve_{\phi}(\bitM) = solve_{op}(Y,\bitM)$ for $Y =
solve_{\psi}(\bitM)$ and

\noindent $\solveApprox_{M_{par}}(\phi, \MGover, \MAunder)= solve_{op}(Y',\MGover)$, where $Y' = \solveApprox_{M_{par}}(\psi, \MGover,\MAunder)$.

Since $Y \subseteq Y'$ from IA, $solve_{op}(Y,\varM)$ is positive monotonic w.r.t. $VB$ and $TB_i$
for $i \in {\GA}$ from Theorem \ref{theorem:monotonicity-of-solve-op}, and $v_{M}[VB] \leq
v_{M_{over}^{\GA}}[VB]$ and $v_{M}[TB_i] \leq v_{M^{\GA}_{over}}[TB_i]$ for $i \in {\GA}$ from Theorem
\ref{theorem:variable-inclusion}, we
have
$solve_{op}(Y,\bitM) \subseteq$ $solve_{op}(Y',\bitM) \subseteq$ $solve_{op}(Y',\MGover)$
and thus $solve_{\phi}(\bitM) \subseteq$
$\solveApprox_{M_{par}}(\phi,\MGover, \MAunder)$.
\end{proof}

%


\section{Satisfiability and Synthesis}
\label{section:SATandSynthesis}

Since the basic predicate $Model_{g,\phi}(\varM)$ is
not monotonic we consider an alternative one:
$\ModelApprox_{g,\phi}(\varM[1],\varM[2])$.

%
For two bit vectors $\bitM[1]$ and $\bitM[2]$ encoding models $M_1$ and $M_2$
compatible with a partial model $M_{par}$, we have:
\begin{center}
$\ModelApprox_{g,\phi}(\bitM[1],\bitM[2])=1$ iff $g \in \solveApprox_{M_{par}}(\phi, \bitM[1], \bitM[2])$.
\end{center}
The following corollary follows directly from Theorem \ref{theorem:solveApprox-is-monotonic}.

\begin{corollary}
$\ModelApprox_{g,\phi}(\varM[1],\varM[2])$ is
\begin{itemize}
\item positive monotonic w.r.t. $TB^{1}_i$ for $i \in {\Ag}$ and $VB^{1}$ and negative monotonic w.r.t. $TB^{2}_i$ for $i \in {\Ag}$ and $VB^{2}$ for
$\phi \in \{p, \neg p, p \wedge q \}$, and

\item positive monotonic w.r.t. $TB^{1}_i$ for $i \in {\GA}$ and $VB^{1}$ and negative monotonic w.r.t. $TB^{2}_i$ for $i \in {\GA}$ and $VB^{2}$ for
$\phi \in \{ \coop{{\GA}}\Next p$, $\coop{{\GA}}\Always p,$ $\coop{{\GA}} p \Until q\} \}$.
\end{itemize}
\end{corollary}
Given a monotonic predicate we can design and apply an efficient
SAT-modulo-ATL solver which uses SAT Modulo Monotonic Theories (SMMT).
This gives us an efficient procedure for ATL satisfiability and synthesis.

The described approach shows that if $M$ is a model of a formula $\phi$,
then the initial state of $M$ belongs to the set of states determined
by $\solveApprox_{M_{par}}(\phi,\MGover, \MAunder)$.
Thus, given an over and under approximation of the set of states satisfying $\phi$,
we can check whether the initial state of $M$ belongs to this approximation.
If not, $M$ is not a model of $\phi$.
Such an approximation can be computed by a partial assignment
built by an SMT solver.
In conclusion, the following theorem follows from Theorems
\ref{theorem:solveApprox-is-monotonic} and \ref{theorem:solveApprox-in-solve}.

\begin{theorem}
\label{theorem:main}
Let $\phi$ be an ATL formula, $M$ be a model with an initial state
$\iota$ such that $M,\iota \models \phi$, and
$M$ is compatible with a partial model $M_{par}$.
Then, we have: $\iota \in \solveApprox_{M_{par}}(\phi, \MAover, \MAunder)$.
\end{theorem}
The following corollary results directly from this theorem.

\begin{corollary}
If $\iota\! \not \in\! \solveApprox_{M_{par}}\!(\phi,\! \MAover,\! \MAunder\!)$, then $M,\!\iota\! \not \models\! \phi$.

If $\solveApprox_{M_{par}}(\phi, \MAover,\MAunder)=\emptyset$, then there is no 
model compatible with $M_{par}$ such that $M,\iota \models \phi$.
\end{corollary}
Now, we are ready to give a procedure for testing satisfiability of the ATL formulae.
Basing on the SMMT framework, we have implemented the
MsAtl tool - a lazy SMT solver for \ATL\ theory.
That is, our implementation exploits a slightly modified MiniSAT\cite{een2003extensible}
as a SAT-solving core, and $\solveApprox$ algorithm as the (main part of the) theory solver for \ATL.
Due to lack of space we are unable to describe our implementation in detail.
However, we sketch below (in a semi-formal way) how our tool works in general.

\textbf{Input:}
(a) an ATL formula $\phi$,
(b) model requirements fixing the number of propositional variables (not less than those appearing in the formula),
the number of agents (not less than those appearing in the formula),
the number of local states for every agent, an initial local state
for every agent, and protocol requirements (if there are any).
The requirements determine a partial model $M_{par}$.

\def\dd{d}
\def\decision{asg}

\textbf{Output:} a model satisfying $\phi$, which meets
the requirements of $M_{par}$ or the answer that such a model does not exist.

Let $\dd$ be an integer variable for tracking the decision depth of the solver,
and $\decision(i)$ denote the variable assigned at the $i$-th step.
\begin{enumerate}

\item Let $\dd:=0$.
\item \label{step1} Compute $\solveApprox_{M_{par}}(\phi,\MAover, \MAunder)$.
\item
If $\iota \in \solveApprox_{M_{par}}(\phi, \MAover, \MAunder)$, then
        \begin{enumerate}
            \item if all variables of $\varM$ are assigned, then return the model.
            \item otherwise: $\dd:=\dd+1$, and SAT-solver core, according to its decision policy,
        assigns a value to the variable $\decision(\dd) \in \varM$.
        In this way the class of the considered models is narrowed down, and the tool
        looks for a valuation which encodes a model satisfying $\phi$.
        Go to step (\ref{step1}).
        \end{enumerate}
\item If $\iota \not \in \solveApprox_{M_{par}}(\phi, \MAover, \MAunder)$, then
        \begin{enumerate}
            \item if $\dd > 0$, then compute conflict clause, analyse conflict, undo recent decisions until appropriate depth $c$, $\dd:=c$,
            assign the opposite value to the variable $\decision(c)$, and go to step (\ref{step1}).
            \item if $\dd = 0$ there is no model which meets the requirements and satisfies $\phi$. Return UNSAT.
        \end{enumerate}
\end{enumerate}

\section{Experimental Results}
\label{section:results}

In order to evaluate the efficiency of our tool we have implemented an ATL formulae generator.
Given the number of agents, groups, and propositional variables, and the depth of the formula,
the generator (using the normal distribution) 
draws a random ATL formula up to the given depth.
We have compared our preliminary results with TATL \cite{david2015deciding} - a tableaux-based tool for ATL satisfiability testing.
Despite the fact that our implementation is at the prototype stage, and there is a lot of space for further
optimizations,\footnote{We plan to increase the efficiency of our tool by introducing
several optimizations, like, e.g., symmetry reductions, formulae caching, and smart clause-learning.}
we have observed several interesting facts.
First of all, for small formulae both tools run rather quickly, in fractions of a second.
When the size of the formula grows, especially when the
number of nested strategy operators increases, the computation time consumed by both tools also grows very quickly.
Moreover, we have found that for unsatisfiable formulae our tool runs quite long, especially
for a large number of states.
This is a typical behaviour for SAT-based methods, which could still be improved by introducing
symmetry reductions preventing the exploration of many isomorphic models.
However, we have found a class of formulae for which our tool outperforms TATL.
These are formulae satisfied by very simple - and often even trivial - models.
%
Table~\ref{tab:exp1} presents the results for a set of such formulae
generated with the following parameter values:
$|\PV| = 3$, 
$|\Ag|= 3$, 
and number of groups equals $4$.
%
%
The table rows have the following meaning (from top to bottom).
The first three rows contain a formula id,
the depth of the formula, i.e., the maximal number of nested strategy operators,
and the total number of Boolean connectives, respectively.
The last two rows present computation times consumed by both tools, in seconds.
The experiments have been performed using a PC equipped with Intel i5-7200U CPU and 16GB RAM running Linux.

\begin{table}[htbp]
\caption{Preliminary experimental results}
\begin{tabular}{|l|r|r|r|r|r|r|r|r|}
\hline
Id & 1 & 2 & 3 & 4 & 5 & 6 & 7 & 8 \\ \hline
Depth & 9 & 13 & 17 & 20 & 23 & 26 & 30 & 33 \\ \hline
Con. & 13 & 19 & 25 & 31 & 35 & 41 & 49 & 55 \\ \hline
MsAtl[s]\!& 0.22 & 0.23 & 0.24 & 0.31 & 0.32 & 0.34 & 0.38 & 0.43 \\ \hline
TATL[s] & 0.58 & 6.2 & 29.7 & 74.6 & 229 & 552 & 1382 & 3948 \\ \hline
\end{tabular}
\label{tab:exp1}
\end{table}
Due to lack of space we do not show here all formulae\footnote{Additional resources, including a prototype
version of our tool, the benchmarks, can be accessed at the (anonymous free hosting) website
\url{http://monosatatl.epizy.com}} but only the shortest one.
The formula~1 of Table~\ref{tab:exp1} is as follows:
$\coop{0} \X (\neg p_0 \lor \coop{1} \G (\neg p_1 \lor \coop{0,1} \F ( \neg p_1 \lor \coop{0,1} \F ( \neg p_0 \lor \coop{2} \F \coop{0} \X (\neg p_0 \lor \coop{1} \G ($ $\neg p_1 \lor \coop{0,1} \G (\coop{0} \F \neg p_0)))))))$.
The subsequent formulae are similar but longer.

It is easy to observe that while scaling the depth of the formulae, the computation time of MonoSatATL grows very slowly,
almost imperceptibly, contrary to TATL for which it increases significantly.


%
%

\section{Conclusions}
\label{section:conclusions}



The paper introduced a new method exploiting SMMT solvers for (bounded) testing of ATL satisfiability
and for constructing (in many cases minimal) ATL models.
Despite the fact that we apply the method to a restricted class of models for ATL under the standard semantics,
our method can be adapted to other classes of multi-agent systems as well as to other ATL semantics including
imperfect information.
Although our implementation is rather at the preliminary stage, the experimental results show a high potential for this approach.





\bibliographystyle{ACM-Reference-Format}  
\bibliography{lit-12-02-2019,bib,wojtek,wojtek-own}  

\end{document}